\pdfoutput=1
\documentclass[runningheads]{llncs}

\usepackage{graphicx, amssymb, amsmath, tikz} 
\def\F{\Bbb F}
\def\P{\Bbb P}
\def\N{\Bbb N}
\def\Z{\Bbb Z}

\def\({\left(}
\def\){\right)}

\def\supp{\text{supp}}

\newcommand{\rmv}[1]{}

\begin{document}
\title{Codes with locality from cyclic extensions of Deligne-Lusztig curves}
%
%
\author{Gretchen L. Matthews \inst{1} \\
and Fernando Pi\~{n}ero
\inst{2}
 }
\authorrunning{G. L. Matthews and F. Pi\~{n}ero}
%
\institute{Department of Mathematics, Virginia Tech, Blacksburg, VA 24061 \\\email{gmatthews@vt.edu}
\url{\empty }\thanks{Partially supported by NSF DMS-1855136.}\\
Department of Mathematics, University of Puerto Rico at Ponce, Ponce, PR \email{fernando.pinero1@upr.edu}\\
\url{\empty}}

\maketitle              
\begin{abstract}
Recently, Skabelund defined new maximal curves which are cyclic extensions of the Suzuki and Ree curves. Previously, the now well-known GK curves were found as cyclic extensions of the Hermitian curve. In this paper, we consider locally recoverable codes constructed from these new curves, complementing that done for the GK curve. Locally recoverable codes allow for the recovery of a single symbol by accessing only a few others which form what is known as a recovery set. If every symbol has at least two disjoint recovery sets, the code is said to have availability. Three constructions are described, as each best fits a particular situation. The first employs the original construction of locally recoverable codes from curves by Tamo and Barg. The second yields codes with availability by appealing to the use of fiber products as described by Haymaker, Malmskog, and Matthews, while the third accomplishes availability by taking products of codes themselves. We see that cyclic extensions of the Deligne-Lusztig curves provide codes with smaller locality than those typically found in the literature. 

\end{abstract}

%
%
%
\section{Introduction}

Maximal curves have played a role in a number of applications in coding theory. For instance, they allow for the construction of long algebraic geometry codes and yield explicit families of codes with parameters exceeding the Gilbert-Varshamov bound \cite{TVZ}. The Deligne-Lusztig curves, which include the Hermitian, Suzuki, and Ree curves, have proven particularly useful. In particular, Hermitian codes are perhaps the best understood algebraic geometry codes other than Reed-Solomon codes. The Suzuki and Ree curves share several important properties with the Hermitian family in that they are optimal with respect to the Hasse-Weil bound and have known automorphism groups; thus, codes from these curves have interesting properties as well. 

More recently, maximal curves have been employed in the construction of codes with locality. In some applications, it is desirable to recover a single (or small number of) codeword symbol(s) by
accessing only a few, say $r$, particular symbols of the received word. This leads to the notion of locally
recoverable codes, or LRCs. Tamo and Barg 
\cite{TB} introduced a construction for codes with
locality that is similar to that of algebraic geometry codes. This motivated much work on locally recoverable codes, including \cite{IPAM}, \cite{BBV}, \cite{Guru}, \cite{Jin}, \cite{LMX}. In \cite{HMM}, we employ maximal curves to construct LRCs with availability $t \geq 2$, meaning each coordinate $j$ has $t$ disjoint recovery sets. Codes with availability make information more available to more users, since recovery of an erasure is not entirely dependent on a single set of coordinates (which may itself contain erasures). 

In this paper, we define codes with locality from new maximal curves constructed by Skabelund \cite{Ska} using cyclic covers of the Suzuki and Ree curves. The Suzuki curve $S_q$ over $\F_q$ gets its name from its automorphism group which is the Suzuki group $Sz(q)$ of order $q^2(q^2 + 1)(q-1)$. In \cite{HS}, Hansen and Stichtenoth considered this curve and applications to algebraic geometry codes leading to other works such as \cite{KP}, \cite{Matthews_Suz}. Recently, Eid, Hammond, Ksir, and Peachey \cite{Eid} constructed an algebraic geometry (AG) code over $\F_{q^4}$ whose automorphism group is $Sz(q)$. Skabelund considers a cyclic extension of $S_q$ and proves it is maximal over $\F_q$ and $\F_{q^4}$. Similarly, the Ree curve $R_q$ over $\F_q$ has a Ree group as its automorphism group. \rmv{In \cite{MTZ}  Montanucci, Timpanella, and Zini construct AG codes over $\F_{q^6}$ using $R_q$ as well as  AG codes from $S_q$.} Both curve constructions are similar to that of the 
Giulietti-Korchm\'{a}ros, or GK, curve, which has already proven useful in constructing codes with locality. These cyclic extensions of the Suzuki and Ree curves have also been utilized for AG codes and for quantum codes from them \cite{MTZ} and their automorphism groups have been determined by Giulietti, Montanucci, Quoos, and Zini \cite{GMQZ}.

This paper is organized as follows. In Section \ref{LRC}, we obtain codes with locality from the cyclic extension $\tilde{S}_q$ of the Suzuki curve $S_q$ and the cyclic extension $\tilde{R}_q$ of the Ree curve $R_q$. The locality is much smaller relative to the alphabet size and code length than comparable constructions. In Section \ref{avail}, we construct codes with availability from $\tilde{S}_q$ and $\tilde{R}_q$. Our constructions build on tools found in \cite{TB} and \cite{HMM}, and some useful background may be found there. Because explicit code descriptions remain out of reach for these standard constructions when employing $R_q$ or $\tilde{R}_q$ (as they depend on explicit bases for Riemann-Roch spaces which remain elusive), we provide an alternate construction for such settings. We also consider constructions from products of codes. In Section \ref{examples}, we consider examples of the above constructions and make some comparisons between them. 

\section{Locally recoverable codes} \label{LRC}
Locally recoverable codes, or LRCs for short, can recover a single (or small number of) codeword symbol(s) by accessing a small number, say $r$, of particular symbols of the received word. In principle, the locality $r$ should be small so as to limit network traffic though this can adversely impact other
code parameters. While an $[n,k,d]$ code $C$, meaning a code of length $n$, dimension $k$, and minimum distance $d$, can recover any $d-1$ erasures or correct any
$\left\lfloor \frac{d-1}{2} \right\rfloor$ errors, this assumes access to all other symbols of the entire received word.  More precisely, the code
$C$ of length $n$ over the alphabet $\F$ (typically taken to be a finite field) is locally recoverable with locality $r$ if and only if for all
$j \in [n]:= \left\{ 1, \dots, n \right\}$ there exists
$$A_j \subseteq [n] \setminus \{ j \} \mathrm{\ with \ } |A_j|=r$$
and $$c_j=\phi_j(c|_{A_j})$$ for some function
$$\phi_j: A_j \rightarrow \F$$ for all $c \in C$. The set $A_j$ is
called a recovery set for the $j$-th coordinate. In this section, we see how cyclic extensions naturally lead to LRCs. 

\subsection{LRCs from cyclic extensions of Suzuki curves}
The Suzuki curve $S_q$ may be described by the equation
$$
S_q: y^q+y=x^{q_0} \left( x^q+x \right)
$$
where $q_0=2^s$, $q=2q_0^2$, and $s \in \N$. It is an optimal curve over $\F_q$, having $q^2+1$ $\F_q$-rational points. Indeed, if $a,b \in \F_q$, $a^q=a$ and $b^q=b$;  since $\mathrm{char \ } \F_q=2$, $b^q+b=0=a^{q_0} \left( a^q+a \right)$. In addition, there is a unique point at infinity $P_{\infty}$ corresponding to $x=z=0$ and $y=1$. The genus of $S_q$ is $q_0 \left(q-1 \right)$ \cite[Lemma 1.9]{HS}. It is maximal over $\F_{q^4}$, having $q^4+1+2q_0q^2(q-1)$ $\F_{q^4}$-rational points \cite[Equation (7)]{Eid}. 
Define 
$$
\tilde{S}_q: \begin{cases} \begin{array}{l} y^q+y=x^{q_0} \left( x^q+x \right) \\ t^m=x^q+x. \end{array} \end{cases}
$$
where $m=q-2q_0+1$. 
The curve $\tilde{S}_q$ has a unique point at infinity, and affine points will be denoted $P_{abc}:=(a:b:c:1)$ to mean the unique zero of $x-a$, $y-b$, and $t-c$, just as those of $S_q$ will be denoted by $P_{ab}$. 
The genus of $\tilde{S}_q$ is $\frac{q^3-2q^2+q}{2}$ \cite{Ska}. According to \cite{Ska}, the number of $\F_{q^4}$-rational points on $\tilde{S}_q$ that are not $\F_q$-rational is 
$$q^5-q^4+q^3-q^2;$$
see also \cite[Section 3]{Ska} for a discussion of the points on this curve. 
Define
$$
\begin{array}{lccc}
g: &\tilde{S}_q &\rightarrow &S_q \\
&P_{abc} &\mapsto &P_{ab} 
\end{array}
$$
Let 
\begin{equation} \label{S1}
S:=S_q \left(\F_{q^4} \right) \setminus S_q \left( \F_q \right).
\end{equation}
Then $|S|=q^4+2q_0q^2(q-1)-q^2$ \cite[Equations (4)-(7)]{Eid}. Set $D:=\sum_{P \in \mathcal{D}} P$ where 
\begin{equation} \label{D1}
\mathcal{D}:=g^{-1} \left( S \right) = \left\{ P_{abc} \in \tilde{S}_q \left(\F_{q^4} \right): c \neq 0 \right\}.
\end{equation}
For each $P_{ab} \in S$, $g^{-1} \left( P_{ab} \right) = \left\{ P_{abc} : c^m=a^q+a \right\}$, so 
\begin{equation} \label{preim}
| g^{-1} \left( P_{ab} \right) | = q-2q_0+1.
\end{equation}

Recall that given a divisor $G$ on a curve $X$ over a field $\F$, the space of functions determined by $G$, sometimes called the Riemann-Roch space of $G$, is 
$$
\mathcal{L}(G):= \left\{ f \in \F(X): (f) \geq -G \right\} \cup \left\{ 0 \right\},$$
where $\F(X)$ denotes the set of rational functions on $X$, and $(f)$ denotes the divisor of the function $f$; to say that $(f)=\sum_{Q \in \mathcal{Z}} a_Q Q - \sum_{P \in \mathcal{P}} b_P P$ with $a_Q, b_P \in \Z^+$ means $f$ has  a zero of order $a_Q$ at $Q$ and a pole of order $b_P$ at $P$. We use the standard notation $(f)_0:=\sum_{Q \in \mathcal{Z}} a_Q Q$ to denote the zero divisor of $f$ and $(f)_{\infty}:=\sum_{P \in \mathcal{P}} b_P P$ to denote the pole divisor of $f$. Let $\alpha \in \Z^+$, and consider the divisor  
$$
G:=\alpha \left( P_{\infty} + \sum_{a,b \in \F_q} P_{ab}  \right)
$$
on $S_q$. It is worth noting that $\mathcal{L} \left(\alpha \left( P_{\infty} + \sum_{a,b \in \F_q} P_{ab}  \right) \right) \cong \mathcal{L} \left( \alpha \left( q^2+1 \right) P_{\infty} \right)$ \cite{Eid}. 
According to \cite[Theorem 1]{Eid}, a basis for $\mathcal{L}(G)$ is given by 
$$
\mathcal{B}:=\left\{ \frac{x^ay^bu^cv^d}{(x^q+x)^e} : \begin{array}{ll} aq+b(q+q_0)+c(q+2q_0)\\ \hspace{.3in} +d(q+2q_0+1) \leq \alpha + eq^2 \\
a \in \left\{ 0, \dots, q-1 \right\}, b \in \left\{ 0, 1 \right\}, \\ c, d \in \left\{ 0, \dots, q_0-1 \right\}, e \in \left\{ 0, \dots, \alpha \right\} \end{array} \right\} \subseteq \F_{q^4} \left( S_q \right)
$$
where $$u=x^{2q_0+1} -y^{2q_0}$$ and $$v= xy^{2q_0} - u^{2q_0}.$$ Set 
$$V:= \left< f t^i : i = 0, \dots, m-2; f \in \mathcal{B} \right>_{\F_{q^4}}.$$
\rmv{where  $M:=dim \mathcal{L}(G)=\alpha \left( q^2+1 \right) - q_0(q-1)+1$. }Now define
$$
\begin{array}{llll}
ev: &V &\rightarrow & \F_{q^4}^{\left( q-2q_0+1 \right) \left( q^4+2q_0q^2 \left( q-1 \right) -q^2 \right)} \\
&f &\mapsto &\left( f \left( P_{abc} \right) \right)_{P_{abc} \in \tilde{S}_q \left( \F_{q^4} \right) \setminus \tilde{S}_q \left( \F_q \right)},
\end{array}
$$
and set $C(D,G,g):=ev(V)$. Note that the evaluation map $ev$ is well-defined, as  $\mid \mathcal D \mid = (q-2q_0+1)(q^4+2q_0q^2(q-1)-q^2)$ and  $f \in V$ has no poles at points in  $\mathcal D$. One may notice that 
\begin{equation}\label{V}
V \subseteq \mathcal L \left( \left(m \alpha + (m-2)q^2 \right) \tilde{P}_{\infty}+ m \alpha \sum_{a, b \in \F_q} P_{a,b,0} \right)
\end{equation} where $\tilde{P}_{\infty}$ 
denotes the unique point of $\tilde{S}_q$ lying above $P_{\infty}$. Let $$G':=\left(m \alpha + (m-2)q^2 \right) \tilde{P}_{\infty}+ m \alpha \sum_{a, b \in \F_q} P_{a,b,0}.$$ We are now ready to state the result.

\begin{theorem} \label{thm}
Suppose $C(D,G,g)$ is constructed as above where $\deg G' < |S|$. Then 
$C(D,G,g)$ is an $[n,k,d]$ code over $\F_{q^4}$ with locality $q-2q_0$, 
$$n=\left( q-2q_0+1 \right) \left( q^4+2q_0q^2 \left( q-1 \right) -q^2 \right),$$ $$k=\left(q-2q_0\right) \left( \alpha \left( q^2+1 \right) - q_0 \left( q-1 \right) +1 \right),$$ and $$d \geq n -\left( m \alpha q^2+m \alpha+(m-2)q^2 \right).$$
\end{theorem}

\begin{proof}
The map $ev$ is injective, since $\deg G' < |S|$ guarantees the kernel of the evaluation map is $\{ 0\}$; this may be observed by noting that if $f \in \ker ev \setminus \{ 0 \}$ then $f$ would have more zeros than poles. Hence, the dimension is given by $\dim_{\F{q^4}}V$ which follows from the facts that $\left\{ t^i : i = 0, 1, \dots, m-1\right\}$ is a basis of $\F_{q^4}(\tilde{S}_q)/ \F_{q^4}({S}_q)$; $\mathcal B$ is a basis for $\F_{q^4}({S}_q)/\F_{q^4}$; and $$\mid \mathcal B \mid = \left( \alpha \left( q^2+1 \right) - q_0 \left( q-1 \right) +1 \right)$$ according to  \cite[Remark 1]{Eid}. We claim that $R:=g^{-1}\left( P_{ab} \right) \setminus \left\{ P_{abc} \right\}$ is a recovery set for the position corresponding to $P_{abc}$. Suppose $f \in V$. Then $$f(x,y,t) = \sum_{i=0}^{m-2} \sum_{j=1}^M a_{ij} f_j^* t^i$$ for some  $a_{ij} \in \F_{q^4}$ and $f_j^* \in \mathcal B$, where $M:= \mid \mathcal B \mid$. Notice that $f(a,b,T) \in \F_{q}\left[T\right]$ and $\deg_T f(a,b,T) \leq m-2$. Hence, $f(a,b,c)$ can be recovered using the $m-1$ interpolation points: $P_{abc'} \in R$. As a result, $f \left( P_{abc} \right)$ may be recovered using only elements of $R$. 

To determine a bound on the minimum distance $d$, we use that $$d \geq wt(ev(h)) \geq n-\deg(h)_{0}$$ where $h=ft^{m-2}$ and $f \in \mathcal{B} \subseteq \mathcal{L}(G)$. Then $\deg(h)_0\geq m \deg (G) + (m-2) q^2 \geq m\alpha(q^2+1)+(m-2)q^2$ as $G$ is a divisor of degree $\alpha(q^2+1)$ on $S_q$, $[\tilde{S}_q: S_q]=m$, and $(t)$ is a divisor on $\tilde{S}_q$ with zero divisor of degree $q^2$. As a result 
$$
d \geq n - \left( m \alpha (q^2+1)+(m-2)q^2 \right).
$$
Alternatively, the bound on the minimum distance may be seen as a consequence of $d \geq n - \deg G'$ using (\ref{V}). 
\end{proof}

\begin{example}
Let $q=8$ and $q_0=2$, so $q^4=4096$. Notice that the Suzuki curve $$S_8: y^8+y=x^2 \left( x^8 + x \right)$$ has $64$ $\F_8$-rational points and $5888$ $\F_{4096}$-rational points. Here, $|S|=5824$ and $n=29120$. Then $C(D,G,g)$ has locality  $4$. We can compare this with an LRC  $C'$ from the Hermitian curve $y^{64}+y=x^{65}$ over the same field, $\F_{4096}$. Using a projection onto the $x$-coordinate gives a code of length $262144$ with locality $63$ whereas projection onto the $y$-coordinate yields locality $64$. Hence, the construction using $\tilde{S}_8$ has a smaller ratios of locality to code length and to alphabet size. 
\end{example}

\begin{remark} \label{rmk}
\begin{enumerate}
\item Other bounds on the minimum distance of the codes in Theorem \ref{thm} may be given; see \cite{TBV} for instance. 
\item 
Alternatively, an LRC may be constructed using the projection 
$$
\begin{array}{lccc}
g: &\tilde{S}_q &\rightarrow &C_m\\
&P_{abc} &\mapsto &Q_{ac} 
\end{array}
$$
where $C_m$ denotes the curve given by $t^m=x^q+x$ and $Q_{ac}$ denotes the common zero of $x-a$ and $t-c$. Let $S$ be as in (\ref{S1}), $D$ as in (\ref{D1}), and  $G':=\alpha Q_{\infty}$  where $Q_{\infty}$ is the point at infinity on $C_m$. Then a basis for $\mathcal{L}\left( \alpha Q_{\infty} \right)$ is given by $$\mathcal{B'}:=\left\{ t^i x^j : i \geq 0, j \in \left\{ 0, \dots, q-1 \right\}, qi+mj \leq \alpha \right\};$$ see, for instance, \cite[Lemma 12.2(i)]{HKT}. Use this to define $$
V = \left< f y^i: i \in \left\{ 0, \dots, q-2 \right\}, f \in \mathcal{B'} \right>.
$$
The code $C(D,G',g)$ has locality $q-1$ and dimension $\left( q-1 \right) |\mathcal{B}|$. \end{enumerate}
\end{remark}

In Section \ref{avail}, we will see how these two approaches can be combined to give LRCs with availability. Before doing so, we turn our attention to cyclic extensions of Ree curves.

\subsection{LRCs from cyclic extensions of Ree curves}

The Ree curve $R_q$ may be described by the equation
$$
R_q: 
\begin{cases}
y^q-y=x^{q_0} \left( x^q-x \right)\\
z^q-z=x^{2q_0} \left( x^q-x \right)\\
\end{cases}
$$
where $q_0=3^s$, $q=3q_0^2$, and $s \in \N$. It is optimal over $\F_{q^6}$. In addition, there is a unique point at infinity. The genus of $R_q$ is $\frac{3}{2}q_0 \left(q-1 \right) \left( q+q_0+1 \right)$ \cite{HS}. Define $$ 
\tilde{R}_q: \begin{cases} \begin{array}{l} y^q-y=x^{q_0} \left( x^q-x \right)\\
z^q-z=x^{2q_0} \left( x^q-x \right) \\ t^m=x^q-x \end{array} \end{cases}
$$
where $m=q-3q_0+1$. 
The curve $\tilde{R}_q$ has a unique point at infinity, and affine points will be denoted $P_{abcd}:=(a:b:c:d:1)$ to mean the unique zero of $x-a$, $y-b$, $z-c$ and $t-d$, just as those of $R_q$ will be denoted by $P_{abc}$. 
The genus of $\tilde{R}_q$ is $\frac{q^4-2q^3+q}{2}$. According to \cite{MTZ}, the number of $\F_{q^6}$-rational points on $\tilde{R}_q$ that are not $\F_q$-rational is 
$$q^7-q^6+q^4-q^3;$$
see also \cite{Ska}. 
Define
$$
\begin{array}{lccc}
g: &\tilde{R}_q &\rightarrow &R_q \\
&P_{abcd} &\mapsto &P_{abc} 
\end{array}
$$
and let 
\begin{equation} \label{S}
S:=R_q \left(\F_{q^6} \right) \setminus R_q \left( \F_q \right).
\end{equation}
\rmv{Then $|S|=q^4+2q_0q^2(q-1)-q^2$ \cite{Eid}. }Set $D:=\sum_{P \in \mathcal{D}} P$ where 
\begin{equation} \label{D}
\mathcal{D}:=g^{-1} \left( S \right) = \left\{ P_{abcd} \in \tilde{R}_q \left(\F_{q^6} \right): d \neq 0 \right\}.
\end{equation}
For each $P_{abc} \in S$, $g^{-1} \left( P_{abc} \right) = \left\{ P_{abcd} : d^m=a^q-a \right\}$, so 
\begin{equation} \label{preim}
| g^{-1} \left( P_{abc} \right) | = q-3q_0+1.
\end{equation}

Consider the divisor  
$
G=\alpha P_{\infty}
$
on $R_q$ with $m \deg G + (m-2) \deg (t)_{\infty} < |S|$. Set 
$$V:= \left< f t^i : i = 0, \dots, m-2; f \in \mathcal{L}(G) \right>_{\F_{q^6}}.$$
\rmv{where  $M:=dim \mathcal{L}(G)=\alpha \left( q^2+1 \right) - q_0(q-1)+1$. }Now define
$$
\begin{array}{llll}
ev: &V &\rightarrow & \F_{q^6}^{|\mathcal D|} \\
&f &\mapsto &\left( f \left( P_{abcd} \right) \right)_{P_{abcd} \in \tilde{R}_q \left( \F_{q^6} \right) \setminus \tilde{R}_q \left( \F_q \right)},
\end{array}
$$
and set $C(D,G,g):=ev(V)$.

\begin{proposition} \label{thm2}
Suppose $C(D,G,g)$ is constructed as above. Then 
$C(D,G,g)$ is an $[q^7-q^6+q^4-q^3,(m-1)\ell(G)]$ code over $\F_{q^6}$ with locality $q-3q_0$.
\end{proposition}

\begin{proof}
This follows similarly to that of Theorem \ref{thm}.
\end{proof}

\begin{remark} \label{ree_rmk} \begin{enumerate} 
\item Explicit bases for $\mathcal{L}(G)$ where $G$ is a divisor on the Ree curve is a topic of current research for arbitrary $q$, even for the case where $G$ is a multiple of the point at infinity. See \cite{Ree} for recent work on related topics. The work \cite{AE} also highlights the challenges of this problem, which was originally stated in \cite{Pedersen}; indeed, when $s=1$ (so $q=27$), the associated Weierstrass semigroup has more than $100$ generators, compared with $2$ in the Hermitian case and $4$ for Suzuki. Hence, the dimension of the codes described in Proposition \ref{thm2} cannot be specified more precisely the expression given above for arbitrary $q$. However, for specific small values of $q$, a set of functions which generate $\mathcal{L}(G)$ may be found computationally. We include this result so that if the theory progresses and sheds more light on this value, LRCs are an immediate consequence. We also note that our interest in the Ree curve is partially motivated by the fact that it allows for results over fields of odd characteristic whose cardinalities are odd powers of primes (unlike the Suzuki curve, which is considered over a field of even characteristic, and the Hermitian curve which is considered over a field with square cardinality). 

\item Also, as in Remark \ref{rmk}, the projection 
$$
\begin{array}{lccc}
g: &\tilde{R}_q &\rightarrow &C_m\\
&P_{abcd} &\mapsto &Q_{ad}, 
\end{array}
$$
where $C_m: t^m=x^q-x$ and $Q_{ad}$ denotes the common zero of $x-a$ and $t-c$, may be used to define a code with different recovery sets than those considered above. 
\item A bound on the minimum distance is given in \cite[Theorem 3.1]{HMM}
\end{enumerate}
\end{remark}

\section{Locally recoverable codes with availability from products} \label{avail}
 
 \subsection{Availability from cyclic extensions viewed as fiber products of curves}
If every coordinate $j$ has $t$ disjoint recovery sets, then $C$ is said to have availability $t$ to
reflect that information is more available to users in the presence of erasure.
In \cite{HMM}, fiber products of curves are used to construct locally recoverable codes with availability. We review the construction in the case $t=2$ below. 

Suppose $X= Y_1 \times_Y Y_2$ where $Y_1$, $Y_2$, and $Y$ are curves over a finite field $\F$ with rational, separable maps $h_i: Y_i \rightarrow Y$. The $\F_q$-rational points of $X$ are $\left\{ \left(P_1,P_2 \right) : P_i \textnormal{ \ is an } \F_q-\textnormal{rational point on } Y_i, h_1(P_1)=h_2(P_2) \right\}$. Thus, there are projection maps $g_i: X \rightarrow Y_i$ defined by $g_i(P_1,P_2)=P_i$; a rational, separable map $g: X \rightarrow Y$ given by $g=h_1 \circ g_1 = h_2 \circ g_2$; maps of function fields $h_i^*:\F(Y) \rightarrow \F(Y_1)$ given by $h_i^*(f):=f \circ h_i$; and primitive elements $x_i$ of the extensions $\F \left(Y_i \right) / h_i^* \left(\F \left( Y \right) \right)$ . Let $S$ be a set of $\F$-rational points on $Y$, and take $D:=\sum_{P \in g^{-1} \left( S \right) } P$.  Choose an effective divisor $G$ on $Y$ of degree $\ell < |S|$, and take a basis $\left\{ f_1, \dots, f_t  \right\}$ for $\mathcal{L}(G)$. Set $$V:=Span \left\{ \left( f_i \circ g \right) x_1^{* e_1} x_2^{* e_2} : 1 \leq i \leq t, 0 \leq e_i \leq \deg h_i - 2 \right\}$$ where $x_i^*=g_i^*(x_i)$ given that $g_i^*:\F(Y_i) \rightarrow \F(X)$ for $i=1,2$. Consider 
$$
\begin{array}{llll}
ev: &V &\rightarrow & \F^n \\
&f &\mapsto &\left( f \left( P_i \right) \right)_{P_i \in \supp D}.
\end{array}
$$
Then the code $C(D,G,g,g_1,g_2):=ev(V)$ has length $|D|=\deg g |S|$, dimension $$t \left( \deg h_1-1 \right) \left( \deg  h_2 -1 \right),$$ and minimum distance bounded below according to \cite{HMM}. For $i=1,2$, $$g_i^{-1} \left( g_i \left( Q\right) \right) \setminus \left\{ Q \right\}$$ serves as a recovery set for $Q \in S$. Hence, $C(D,G,g,g_1,g_2)$ has locality $2$. Next we apply this construction to $\tilde{S}_q$ and $\tilde{R}_q$. 

\subsubsection{Cyclic extensions of Suzuki curves as fiber products.}

Because 
$\tilde{S}_q$ is the fiber product of covers $S_q \rightarrow \P_x^1$ and $C_m \rightarrow \P_x^1$, we may apply the construction to obtain a code with availability $2$ and localities $m-1$ and $q-1$; that is, every coordinate has 2 disjoint recovery sets, one of cardinality $q-2q_0$ and one of cardinality $q-1$. To do this, consider the projection maps $g_1: \tilde{S}_q \rightarrow C_m$, $g_2: \tilde{S}_q \rightarrow S_q$, and $g: \tilde{S}_q \rightarrow \P_x^1$. We take $S$ as in (\ref{S1}), $D$ as in (\ref{D1}), and $G:=\alpha P_{\infty}$ where $P_{\infty}$ is the unique point at infinity on $\P_x^1$. Fix a basis $\mathcal{B}$ of $\mathcal{L}(G)$, and  
$$V:= \left< f y^i t^j : 0 \leq i \leq q-2, 0 \leq j \leq m-2, f \in \mathcal{B} \right>_{\F_{q^4}}.$$

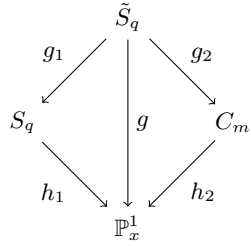
\begin{figure}  
\centering
\begin{tikzpicture}[node distance = 2cm, auto]
  \node (Y) {$\P_x^1$};
  \node (Y1) [node distance=1.4cm, left of=Y, above of=Y] {$S_q$};
   \node (Y2) [node distance=1.4cm, right of=Y, above of=Y] {$C_m$};
   \node (X) [node distance=1.4cm, left of=Y2, above of=Y2] {$\tilde{S}_q$};
  \draw[->] (Y1) to node [swap] {$h_1$} (Y);
  \draw[->] (Y2) to node {$h_2$} (Y);
  \draw[->] (X) to node [swap] {$g_1$} (Y1);
  \draw[->] (X) to node {$g_2$} (Y2);
  \draw[->] (X) to node {$g$} (Y);
 \end{tikzpicture}
   \caption{Cyclic extension of Suzuki curve viewed as a fiber product}
\label{Suzuki_fig}
\end{figure}

\begin{theorem} \label{thm3}
Suppose $C(D,G,g,g_1,g_2)$ is constructed as above. Then the code
$C(D,G,g,g_1,g_2)$ is an $[n,k,d]$ code over $\F_{q^4}$ with availability $2$ and recovery sets for each coordinate of sizes 
$q-2q_0$ and $q-1$, where  
$$n=\left( q-2q_0+1 \right) \left( q^4+2q_0q^2 \left( q-1 \right) -q^2 \right),$$ $$k=\left(q-2q_0\right) \left( \alpha +1 \right) \left( q-1 \right),$$ and $$d \geq n - \left( \alpha m q +(q-2)m(q+q_0)+(m-2)q^2\right).$$
\end{theorem}

\begin{proof}
The length and dimension can be verified directly by applying  \cite[Theorem 3.1]{HMM}. To determine the minimum distance $d$, we use the fact that $d \geq n-wt(ev(h)) \geq n-\deg (h)_0$ where $h=fy^{q-2}t^{m-2}$ and $f \in \mathcal{L}(\alpha P_{\infty})$. Then $(h)=(f)+(q-2)(y)+(m-2)(t)$. Note that when considered as a functions on $\tilde{S}_q$, $\deg (f)_0 \leq \alpha m q$, $\deg (y)_0 \leq m(q+q_0)$, and $\deg (t)_0 \leq q^2$. Putting this together, we conclude that $d \geq n - \left( \alpha m q +(q-2)m(q+q_0)+(m-2)q^2\right)$, which coincides with that given in \cite[Theorem 3.1]{HMM}.

We claim that $$R^{(1)}:=g_2^{-1} \left( g_2 \left(P_{abc} \right) \right)  \setminus \left\{ P_{abc} \right\}= \left\{ P_{ab'c} : b' \in \F_{q^4} \setminus \left\{ b \right\} \right\}$$  and $$R^{(2)}:=g_1^{-1} \left( g_1 \left(P_{abc} \right) \right) \setminus \left\{ P_{abc} \right\} = \left\{ P_{abc'} : c' \in \F_{q^4}  \setminus \left\{ c \right\} \right\}$$
are  recovery sets for the position corresponding to $P_{abc}$. Suppose $f \in V$. Then $f(x,y,t) = \sum_{i=0}^{m-2} \sum_{j=1}^M a_{ij} f_j^* t^i$. Notice that $f(a,b,T) \in \F_{q}\left[T \right]$ and the degree is bounded by $\deg_T f(a,b,T) \leq m-2$. Hence, $f(a,b,c)$ can be recovered using the $m-1$ interpolation points: $P_{abc'} \in R$. As a result, $f \left( P_{abc} \right)$ may be recovered using only elements of $R$.
\end{proof}

Observe that the functions in the set $V$ are modified from the construction in Section \ref{LRC} in order to obtain multiple recovery sets for each position, thus impacting the dimension of the code. 

One might compare this with the code found in \cite[Theorem 6.1]{HMM}, which has availability $2$ with recovery sets of size $q-1$, length  $n = q(q -1)(q^2 + 2qq_0 + q + 1)$ and dimension $k =(q -1)(q - 2)(q^2 + 2qq_0 + q + 1)$. Notice that the new codes defined using $\tilde{S}_q$ give the option of using a smaller recovery set (cardinality $q-2q_0$ compared with $q-1$). 

\subsubsection{Cyclic extensions of Ree curves as fiber products.}

Because 
$\tilde{R}_q$ is the fiber product of  $R_q \rightarrow \P_x^1$ and $C_m \rightarrow \P_x^1$, we may apply this construction to obtain a code with availability $2$ and localities $m-1$ and $q-1$; that is, every coordinate has 2 disjoint recovery sets, one of cardinality $q-3q_0$ and one of cardinality $q-1$. To do this, consider the projection maps $g_1: \tilde{R}_q \rightarrow C_m$, $g_2: \tilde{R}_q \rightarrow R_q$, and $g: \tilde{R}_q \rightarrow \P_x^1$. We take $S$ as in (\ref{S}), $D$ as in (\ref{D}), and $G$ is a divisor on $\P_x^1$. Fix a basis $\mathcal{B}$ of $\mathcal{L}(G)$, and  
$$V:= \left< f y^i t^j :  0 \leq i \leq q-2, 0 \leq j \leq q-3q_0-1, f \in \mathcal{B} \right>_{\F_{q^6}}. $$

\begin{figure}  
\centering
\begin{tikzpicture}[node distance = 2cm, auto]
  \node (Y) {$\P_x^1$};
  \node (Y1) [node distance=1.4cm, left of=Y, above of=Y] {$R_q$};
   \node (Y2) [node distance=1.4cm, right of=Y, above of=Y] {$C_m$};
   \node (X) [node distance=1.4cm, left of=Y2, above of=Y2] {$\tilde{R}_q$};
  \draw[->] (Y1) to node [swap] {$h_1$} (Y);
  \draw[->] (Y2) to node {$h_2$} (Y);
  \draw[->] (X) to node [swap] {$g_1$} (Y1);
  \draw[->] (X) to node {$g_2$} (Y2);
  \draw[->] (X) to node {$g$} (Y);
 \end{tikzpicture}
   \caption{Cyclic extension of Ree curve viewed as a fiber product}
\label{Suzuki_fig}
\end{figure}

\begin{proposition}
The code $C(D,G,g,g_1,g_2)$  constructed as above is a code with parameters $[q^7-q^6+q^4-q^3,\ell(G)(q-1)(q-3q_0)]$ code over $\F_{q^6}$ with availability $2$ and recovery sets for each coordinate of sizes 
$q-3q_0$ and $q-1$.
\end{proposition}

\begin{proof}
The proof is similar to that of Theorem \ref{thm3}. 
\end{proof}

\begin{remark} \label{ree_rmk2}
\begin{enumerate}
\item As noted in Remark \ref{ree_rmk}, the explicit construction for codes from the Ree curve depends on that of bases for certain Riemann-Roch spaces. We provide an alternate LRC with availability construction from the Ree curve in the Section \ref{examples}. There, we see codes with more accessible parameters due to choosing functions to evaluate carefully, rather than beginning with an entire Riemann-Roch space which is difficult to describe.
\item A bound on the minimum distance is given in \cite[Theorem 3.1]{HMM}.
\end{enumerate}
\end{remark}

Observe the functions in the set $V$ are modified from the construction in Section \ref{LRC} in order to obtain multiple recovery sets for each position, thus impacting the dimension of the code.

\subsection{Availability from products of codes}
	We may also take products of codes themselves to obtain LRCs with availability, as detailed below. 
We begin with the simplest definition, the product of two codes, $C_1$ and $C_2$, which may be generalized to more factors. Examples of this construction may be found in the next section.  

\begin{definition}
Let $C_1$ be an $[n_1, k_1, d_1]$ code and $C_2$ be an $[n_2, k_2, d_2]$ code over the same alphabet $\F$. The product code of $C_1$ and $C_2$ is defined by assigning symbols from $\F$ to the pairs $(i,j) \in [n_1] \times [n_2]$ such that the symbols assigned in $[n_1] \times \{j\}$, for $j \in [n_2]$ are a codeword in $C_1$ and $\{i\} \times [n_2]$ for $i \in [n_1]$ are a codeword in $C_2$; that is, 
$$ C_1 \times C_2  :=  \{  (a_i b_j) \in \F^{[n_1]\times[n_2]} \ | \ (a_{1}, a_{2}, \ldots, a_{n_1}) \in C_1, (b_{1}, b_{2}, \ldots, b_{ n_2}) \in C_2  \} $$ 
\end{definition}

An alternative definition is to place symbols from $\F$ in an $n_1 \times n_2$ rectangular array such that each column is a codeword of $C_1$ and each row is a codeword of $C_2$. See also \cite{LMS}. 

\begin{theorem}
Let $C_1$ be an $[n_1, k_1, d_1]$ code and  $C_2$ be an $[n_2, k_2, d_2]$ code. Then the code $C_1 \times C_2$ is a $[n_1n_2, k_1k_2, d_1d_2]$ code with availability $2$. Moreover, if $C_1$ has locality $r_1$ and availability $l_1$ and $C_2$ has locality $r_2$ and availability $l_2$, then $C_1 \times C_2$ is a code of availability $l_1+l_2$ and locality $r_1+r_2$.
\end{theorem}
\begin{proof} 
Let $D$ denote the minimum distance of the code $C_1 \times C_2$. If $(i,j)$ is a nonzero position, then there are $d_1$ positions in the set $[n_1] \times \{j\}$ which have a nonzero entry. Suppose those nonzero positions are $(i_1, j), (i_2, j), \ldots, (i_{d_1}, j)$. For each of those nonzero positions $(i_s, j)$, there are $d_2$ nonzero positions in $\{i_s\} \times [n_2]$. Thus there are at least $d_1d_2$ nonzero positions.

In order to prove equality, let $(a_1, a_2, \ldots, a_{n_1} )$ be a codeword of weight $d_1$ in $C_1$, and $(b_1, b_2, \ldots, b_{n_2})$ be a codeword of weight $d_2$ in $C_2$. Then the codeword defined by $c_{i,j} = a_ib_j$ is the required codeword of weight $d_1d_2$.

Let $I_1$ be an information set for $C_1$ and let $I_2$ be an information set of $C_2$. The $i^{th}$ coordinate of $c \in C_1$ may be written as the linear combination $\mathbf{p_i}\mathbf{m_1}$ for a message vector $m_1 \in \F^{I_1}$. Likewise, the $i^{th}$ coordinate of $c \in C_1$ may be written as the linear combination $\mathbf{q_j}\mathbf{m_2}$ for a message vector $m_2 \in \F^{I_2}$. After placing any values in $I_1 \times I_2$, the remaining values are given by $c_{i,j} = \mathbf{p_i}\mathbf{m_1}\mathbf{q_j}\mathbf{m_2}$.

Note that position $(i,j)$ is in the two sets $[n_1] \times \{j\}$ and $\{i\} \times [n_2]$. These two sets have only $(i,j)$ in common. Thus, $[n_1]\times \{j\} \setminus (i,j)$ and $\{i\} \times [n_2] \setminus (i,j)$ are recovery sets for $(i,j)$; note that they are disjoint as required for availability. 

Consider $(i,j) \in [n_1] \times [n_2]$. As $C_1$ is a code of availability $l_1$ there are $l_1$ disjoint sets , $I_1, I_,2, \ldots, I_{l_1}$ in $[n_1]\setminus \{i\}$ from which position $i$ may be recovered. Likewise as $C_2$ is a code of availability $l_2$ there are $l_2$ disjoint sets , $J_1,J_2, \ldots, J_{l_2}$ in $[n_2]\setminus \{j\}$ from which position $j$ may be recovered. The sets $I_1 \times \{j \}$,$I_2 \times \{j \}, \ldots$, $I_{l_1} \times \{j \}$, $\{i\} \times J_1$,$\{i\} \times J_2, \ldots$, $\{i\} \times J_{l_2}$ are then $l_1+l_2$ recovery sets which are disjoint; this gives the desired availability. \end{proof}

\section{Examples} \label{examples}

A number of examples of LRCs are given in this section, and some comparisons are drawn between instances of the constructions discussed in this paper as well as those appearing elsewhere in the literature. In addition, we provide LRCs on the Ree curve via a construction that allows for computable parameters despite the issues mentioned in Remarks \ref{ree_rmk} and \ref{ree_rmk2}.

Tamo and Barg gave a seminal construction of an optimal LRC code of locality $r$ in \cite{TB}. The LRC construction is based on a set $L \subseteq  \F_q$, a partition of $L$ into disjoint subsets $A_1$, $A_2$, ... , $A_m$ where each set $A_i$ has size $r+1$ and a polynomial $g(x)$ of degree $r+1$ such that $g$ is constant on each subset $A_i$. Tamo and Barg construct an LRC code from a subcode of the Reed--Solomon code over $L$ of dimension $k'$ by evaluating the functions of the form $X^i g(X)^j$ where $0 \leq i \leq r, i \neq s$ for a fixed $0 \leq s \leq r$ and $i + (r+1)j \leq k'-1$. There are many partitions and many choices for $g(X)$. However, we shall focus on partitions given by linear subsets of $\F_q$ or by cosets of the multiplicative group of $\F_q$. We shall use evaluation codes as a generalization of Reed--Solomon codes and AG codes.

Let $A = \{\alpha_1, \alpha_2, \ldots, \alpha_n\} \subseteq \mathbb{F}_q^m$. Let $f(x_1, x_2, \ldots, x_m)$ be a polynomial in $m$ variables. The evaluation map of $f$ on $A$ is defined as $$ev_A : \mathbb{F}_q[x_1, x_2, \ldots, x_m] \rightarrow \mathbb{F}_q^n $$  where $$ev_A(f) = (f(\alpha_1), f(\alpha_2), \ldots, f(\alpha_n)).$$ 
We remark that the vanishing ideal of $A$, namely $$I_A =\{f \in \mathbb{F}_q[x_1, x_2, \ldots, x_m] \ | f(\alpha) = 0 \forall \alpha \in A  \},$$ is the kernel of the evaluation map $ev_A$.

Let $A= \{\alpha_1, \alpha_2, \ldots, \alpha_n\} \subseteq \mathbb{F}_q^m$. Let $L$ be a subspace of $\mathbb{F}_q[x_1, x_2, \ldots, x_m]$. The set $$C(A, L) = \{ ev_A(f) \ | \ f \in L \} $$ is known as an affine variety code. The definition of an affine variety code simply states that a linear code may be constructed  by evaluating functions on a set of points. In most cases, the structure of $L$ or $A$ will imply certain properties of the code hold,  such as dimension, minimum distance or locality.

\begin{lemma}
\label{lemma:ProductLRC}
Let $V_1 \subseteq \F_q^{m_1}$. Let $L_1$ be a subspace of $\mathbb{F}_q[x_1, x_2, \ldots, x_{m_1}]$. Similarly, take $V_2 \subseteq \F_q^{m_2}$ and $L_2$ be a subspace of $\mathbb{F}_q[y_1,y_2, \ldots, y_{m_2}]$. Consider the evaluation codes: $C_1 = C(V_1, L_1)$ and $C_2 = C(V_2, L_2)$.  
The product code $C_1 \times C_2$ is the evaluation code $C_3 = C(V_3, L_3)$, where $V_3 = V_1 \times V_2 \subseteq \F_q^{m_1+m_2}$ and the set of evaluated functions is $L_3 = \{ f(X)g(Y), f \in L_1, g \in L_2 \}$.
\end{lemma}
\begin{proof} It is clear that taking $f \in L_1$ and $g \in L_2$  and evaluating the product $f(X)g(Y)$  on the array $\{(\alpha, \beta) \ | \ \alpha \in V_1, \beta \in V_2  \}$ will give a codeword of the form $(f(\alpha)g(\beta))$. This codeword is also a codeword of $C_1 \times C_2$.  

In order to prove equality, we use a dimensional analysis. As $C_1$ is a code of dimension $k_1$, there exist $f_1, f_2, \ldots, f_{k_1}$ functions of $L_1$ and $\alpha_1, \alpha_2, \ldots, \alpha_{k_1} \in V_1$ such that $f_i(\alpha_j) = \delta_{i,j}$. Likewise, there exist $g_1, g_2, \ldots, g_{k_2}$ functions of $L_2$ and $\beta_1, \beta_2, \ldots, \beta_{k_2} \in V_2$ such that $g_{i'}(\beta_{j'}) = \delta_{i',j'}$. The evaluation of the functions $f_i g_{i'}$ on the points $\alpha_j \beta_{j'}$ will also imply the image has dimension $k_1k_2$.
 \end{proof}

We will construct LRC codes based on the product code of the Tamo--Barg construction and Lemma \ref{lemma:ProductLRC}. In particular, we shall take $V_1 \times  V_2, V_1 \times V_2 \times  V_3$, for $V_i \subseteq \F_q$ as our evaluation sets and the evaluation functions to be $$\{ f_1(X)f_2(Y)f_3(Z) \ | \ f_i \in   L_i\}$$ where $L_i = \{ T^ag_i(T)^j,  0 \leq a \leq r_i, a \neq s_i, a+ (r_i+1)j < k'\}$.

Note that classical Hermitian codes are obtained by evaluating monomials of the form $\mathcal{M}(s):= \{X^iY^j \ | \ iq+j(q+1) \leq s  \}$ on the $q^3$ points of the form $A = \{ (\alpha, \beta) \in \F_{q^2}^2 \ | \ \alpha^{q+1} = \beta^q+\beta, \alpha \neq 0 \}$. However, as the vanishing ideal of $A$ is the ideal spanned by $X^{q+1}-Y^q-Y$ and $X^{q^2-1}-1$, we may consider the Hermitian code obtained by evaluating the functions $$\mathcal{M}(s)= \{X^iY^j \ | \ iq+j(q+1) \leq s,   0 \leq i \leq q^2-2, 0 \leq j \leq q-1\}.$$ 

In order to find a subcode of the Hermitian code with a given locality, proceed as follows: Let $g_1(x)$ be a polynomial of degree $r_1+1$. Let $A_1, A_2, \ldots, A_{\frac{q^2-1}{r_1+1}}$ be a partition of $\F_{q^2}$ into multiplicative cosets of $\F_{q^2}^*$ where $g_1(X)$ is constant on each $A_i$.  Likewise let $g_2(Y)$ be a polynomial of degree $r_2+1$. Let $B_1, B_2, \ldots, B_{\frac{q}{r_2+1}}$ be a partitiion of $\{\gamma | \gamma^q+\gamma = 0  \}$ where $g_2(Y)$ is constant on each $B_j$. For fixed $s_1, s_2$, the code obtained by evaluating 
$$
L_{s1, s2}(s) = \left\{ X^{i_1}g_1(X)^{i_2}Y^{j_1}g_2(Y)^{j_2}  \ | \ 
\begin{array}{l}
0 \leq i_1 \leq r_1, i_1 \neq s_1, 0 \leq j_1 \leq r_2,  \\0 \leq i_1+(r_1+1)i_2 \leq q^2-2, \\0 \leq j_1+(r_2+1)j_2 \leq q-1 \end{array}  \right\}$$ is a subcode of the Hermitian code of degree $s$ which also has locality $r_1$ and locality $r_2$ with availability $2$.
In this case we obtain the codes over $\F_{16}$ with locality $3$ and availability $2$ and the following parameters: $[64, 1, 64]$,
$[64, 2, 60]$, $[64, 3, 59]$, $[64,\  4, 56]$,
$[64, \ 5, 55]$, $[64,\  6, 54]$, 
$[64, 7, 51]$, $[64, 8, 50]$, 
$[64, 9, 48]$, $[64, 10, 46]$, 
$[64, 11, 44]$, $[64, 12, 43]$, 
$[64, 13, 40]$, $[64, 14, 39]$, 
$[64, 15, 38]$, $[64, 16, 35]$, 
$[64, 17, 34]$, $[64, 18, 32]$, 
$[64, 19, 30]$, $[64, 20, 28]$, 
$[64, 21, 27]$, $[64, 22, 24]$, 
$[64, 23, 23]$, $[64, 24, 22]$, 
$[64, 25, 19]$, $[64, 26, 18]$, 
$[64, 27, 16]$, $[64, 28, 14]$, 
$[64, 29, 12]$, $[64, 30, 12]$,
$[64, 31, 8]$, $[64, 32, 8]$, 
$[64, 33, 8]$, $[64, 34, 6]$, 
$[64, 35, 6]$, $[64, 36, 4]$.
Note these codes handily outperform the product code construction from two Reed--Solomon code over $\F_{16}$, which have parameters: $[64, 3, 56]$, $[64, 6, 42]$, $[64, 9, 28]$,$[64, 12, 30]$, $[64, 18, 20]$, $[64, 27, 12]$.

For Suzuki curves, we use two different constructions of LRCs. As the full affine plane $\F_q^2$ is the set of $\F_q$--rational points, we evaluate the intersection of a product code of two LRC codes with the Suzuki code of length $q^2$. The vanishing ideal of the full affine plane is the ideal spanned by $X^q-X, Y^q-Y$. In this case, the Suzuki code is obtained by taking the polynomials of low order at infinity and evaluating at the $q^2$ rational points. One can also take the remainders modulo $X^q+X, Y^q+Y$ to determine the dimension of the code instead. Hence the Suzuki code is obtained by evaluating $L(s) =$ $$ \{ X^aY^bU^cV^d\mod X^q+X, Y^q+Y\ | \ aq + b(q+q_0)+ c(q+2q_0)+ d(q+2q_0+1) \leq s \}.$$ 

\begin{table}  \label{suz_table}
\centering
\begin{tabular} {|c| c| c|}
\hline
$[s_1, s_2]$ & Suzuki code & Suzuki code with locality and availability \\ \hline
$[1, 0]$ & $[64, 1, 64]$ & $[64, 1, 64]$ \\
$[1, 0]$ & $[64, 2, 56]$ & $[64, 2, 56]$ \\
$[2, 0]$ & $[64, 3, 54]$ & $[64, 3, 54]$ \\
$[3, 0]$ & $[64, 5, 51]$ & $[64, 4, 51]$ \\
$[3, 0]$ & $[64, 6, 48]$ & $[64, 5, 48]$ \\
$[3, 0]$ & $[64, 7, 46]$ & $[64, 6, 46]$ \\
$[3, 0]$ & $[64, 8, 44]$ & $[64, 7, 44]$ \\
$[3, 0]$ & $[64, 12, 40]$ & $[64, 8, 40]$ \\
$[3, 0]$ & $[64, 14, 38]$ & $[64, 9, 38]$ \\
$[3, 0]$ & $[64, 15, 36]$ & $[64, 10, 36]$ \\
$[2, 0]$ & $[64, 19, 32]$ & $[64, 11, 32]$ \\
$[3, 0]$ & $[64, 20, 31]$ & $[64, 12, 31]$ \\
$[3, 0]$ & $[64, 21, 30]$ & $[64, 13, 30]$ \\
$[3, 0]$ & $[64, 23, 28]$ & $[64, 14, 28]$ \\
$[3, 0]$ & $[64, 24, 27]$ & $[64, 15, 27]$ \\
$[3, 0]$ & $[64, 25, 26]$ & $[64, 16, 26]$ \\
$[3, 0]$ & $[64, 27, 24]$ & $[64, 17, 24]$ \\
$[3, 0]$ & $[64, 29, 22]$ & $[64, 18, 22]$ \\
$[3, 0]$ & $[64, 31, 20]$ & $[64, 19, 20]$ \\
$[3, 0]$ & $[64, 32, 19]$ & $[64, 20, 19]$ \\
$[3, 0]$ & $[64, 35, 16]$ & $[64, 21, 16]$ \\
$[3, 0]$ & $[64, 37, 14]$ & $[64, 22, 14]$ \\
$[3, 0]$ & $[64, 39, 12]$ & $[64, 23, 12]$ \\
$[3, 0]$ & $[64, 40, 11]$ & $[64, 24, 11]$ \\
$[3, 0]$ & $[64, 43, 8]$ & $[64, 25, 8]$ \\
$[3, 0]$ & $[64, 44, 7]$ & $[64, 26, 7]$ \\
$[3, 0]$ & $[64, 45, 6]$ & $[64, 27, 6]$ \\
$[3, 0]$ & $[64, 47, 4]$ & $[64, 28, 4]$ \\
$[3, 0]$ & $[64, 49, 2]$ & $[64, 29, 2]$ \\
\hline
\end{tabular}
\caption{Comparison of parameters of codes from the Suzuki curve using a standard AG code construction and those with locality and availability}
\end{table}

In this case, the optimal LRC codes of locality $3$ and length $8$ have parameters: $[8, 1, 8]$, $[8, 2, 7]$,$[8, 3, 6]$, $[8, 4, 4]$, $[8, 5, 3]$, $[8, 6, 2]$.
In order to get a product code of locality $3$ and availability $2$ from these codes, we get a $[64,10, 21]$ code. From the Suzuki code construction, after imposing additional LRC conditions, we get a $[64,10,36]$ code with the same locality and availability parameters.

In the following table, we compare some Suzuki LRC codes with some RS product LRC codes. Both have the same length, symbols,  locality and availability. Note that we were able to improve on most of the Product code constructions, except for $[64, 25, 9]$. We expect to improve our codes by improving the minimum distance bounds of the Suzuki codes.

\begin{table}
 \label{suz_table_comp}
\centering
\begin{tabular} {|c| c|}
\hline
Suzuki code with locality and availability  & Comparable Product code \\ \hline
 $[64, 1, 64]$  & $[64, 1, 64]$ \\
 $[64, 2, 56]$  & $[64, 2, 56]$ \\
 $[64, 3, 54]$  & $[64, 3, 48]$ \\
 $[64, 4, 51]$  & $[64, 4, 49]$ \\
 $[64, 5, 48]$ & $[64, 5, 24]$\\
 $[64, 6, 46]$ & $[64, 6, 42]$ \\
 $[64, 7, 44]$ & $[64, 6, 42]$ \\
 $[64, 8, 40]$  & $[64, 8, 28]$ \\
 $[64, 9, 38]$  & $[64, 9, 36]$ \\
 $[64, 10, 36]$  & $[64, 10, 21]$ \\
 $[64, 11, 32]$ & $[64, 10, 21]$ \\
 $[64, 12, 31]$ & $[64, 12, 24]$ \\
 $[64, 13, 30]$ & $[64, 12, 24]$ \\
 $[64, 14, 28]$ & $[64, 12, 24]$ \\
 $[64, 15, 27]$ & $[64, 15, 18]$ \\
 $[64, 16, 26]$  & $[64, 16, 16]$ \\
 $[64, 17, 24]$ & $[64, 16, 16]$ \\
 $[64, 18, 22]$ & $[64, 18, 12]$ \\
 $[64, 19, 20]$ & $[64, 18, 12]$ \\
 $[64, 20, 19]$ & $[64, 20, 12]$ \\
 $[64, 21, 16]$ & $[64, 20, 12]$ \\
 $[64, 22, 14]$ & $[64, 20, 12]$ \\
 $[64, 23, 12]$ & $[64, 20, 12]$ \\
$[64, 24, 11]$  & $[64, 24, 8]$ \\
 $[64, 25, 8]$ & $[64, 25, 9]$ \\
 $[64, 26, 7]$ & $[64, 25, 9]$ \\
$[64, 27, 6]$ & $[64, 25, 9]$ \\
 $[64, 28, 4]$ & $[64, 25, 9]$ \\
 $[64, 29, 2]$ & $[64, 25, 9]$ \\ \hline
\end{tabular}
\caption{Comparison of parameters of codes from the Suzuki curve with locality and availability and product codes}
\end{table}

To get LRCs from the Ree curve we shall make a similar construction to the codes from the Suzuki curve. Due to the abundance of possible codes and availabilities, we shall restrict ourselves to the case where $q = 27$, availability is $3$ and $r_1=r_2 =r_3 = 8$. Please recall that the Ree curve $R_q$ may be described by the equation
$$
R_q: 
\begin{cases}
y^q-y=x^{q_0} \left( x^q-x \right)\\
z^q-z=x^{2q_0} \left( x^q-x \right)\\
\end{cases}
$$
where $q_0=3^s$, $q=3q_0^2$, and $s \in \N$. The valuation of $x$ at infinity is $q$, the valuation of $y$ at infinity is $q+q_0$ and the valuation of $z$ at infinity is $q+2q_0$.  If $G$ represents the pole at infinity of the Ree curve, and $D$ is the divisor corresponding to the $\mathbf{F}_{27}$--affine points of the Ree curve, then the code $C(sG, D)$ is the algebraic geometry code obtained by evaluating all functions having poles only at infinity of order $\leq s$. We shall compare LRC subcodes of $C(sG,D)$ with product codes of Reed--Solomon codes. We shall use a particular Tamo--Barg construction \cite{TB} for $\mathbf{F}_{27}$. In this case, our sets will be places where the trace is constant $A_0 = \{ a \in \mathbf{F}_{27} \ | \ a+a^3+a^9 = 0 \}$, $A_1 = \{ a \in \mathbf{F}_{27} \ | \ a+a^3+a^9 = 1\}$ and $A_2 = \{ a \in \mathbf{F}_{27} \ | \ a+a^3+a^9 = 2 \}$. The polynomial $L(T) = T+T^3+T^9$ is constant on the recovery sets $A_0$, $A_1$, $A_2$ and the evaluation of the polynomials in $\{ L(T)^{i}T^j \ | \ 0 \leq i \leq 2, 0 \leq j \leq 7, 2i+j \leq k \}$ gives a subcode of the Reed--Solomon code $RS_{27}(\mathbf{F}_{27}, k)$ which is an LRC with locality $8$.

The possible Reed--Solomon codes with locality $8$ and length $27$ of this form are: $[27, 1, 27]$,
 $[27, 2, 26]$,
 $[27, 3, 25]$,
 $[27, 4, 24]$,
 $[27, 5, 23]$,
 $[27, 6, 22]$,
 $[27, 7, 21]$,
 $[27, 8, 20]$,
 $[27, 9, 18]$,
 $[27, 10, 17]$,
 $[27, 11, 16]$,
 $[27, 12, 15]$,
 $[27, 13, 14]$,
 $[27, 14, 13]$,
 $[27, 15, 12]$,
 $[27, 16, 11]$,
 $[27, 17, 9]$,
 $[27, 18, 8]$,
 $[27, 19, 7]$,
 $[27, 20, 6]$,
 $[27, 21, 5]$, \\
 $[27, 22, 4]$,
 $[27, 23, 3]$,
 $[27, 24, 2]$.

The key idea of this proof is that both the product code of Reed--Solomon codes and the subcodes of the Ree AG code may be considered as evaluation codes of combinations of monomials in $\mathcal{L} = \{ X^iY^jZ^l \ | \ 0 \leq i,j,l \leq 26\}$. In the $\mathbf{F}_{27}^3$.

The subcode of $C(sG, D)$ is obtained by evaluating the monomials in 
$$\mathcal{L}(s) = \{ X^aY^bZ^c \in \mathcal{L} \ | \ aq+b(q+q_0)+c(q+2q_0) \leq s, 0 \leq a,b,c \leq q-1 \}.$$ 
The subcode of $C(sG, D)$ with locality $8$ and availability $3$ is given by evaluating polynomials of the form
$L(X)^{a1}X^{a_2}L(Y)^{b1}Y^{b2}L(Z)^{c1}Z^{c2}$ where $0 \leq a_1, b_1, c_1 \leq 2, 0 \leq a_2, b_2, c_2 \leq 7$ and subject to the degree constrain that the polynomials should also be in  $\mathcal{L}(s)$.
Note that depending on the parameters of the codes we might find better Reed--Solomon product codes as LRCs or better LRC subcodes from the Ree curve.

For example, comparing codes with minimum distance $600$ we get an LRC with parameters   $[19683,4536,600]$, locality $8$ and availability $3$ from the product code construction of the Reed--Solomon code and a $[19683,2937,600]$ LRC with locality $8$ and availability $3$ from $C(sG, D)$. However, comparing codes with dimension $200$ we get a   $[19683,200,10580]$ LRC with locality $8$ and availability $3$ from the product code construction of the Reed--Solomon code and a $[19683,201, 13086]$ LRC with locality $8$ and availability $3$ from the AG code.

There is also an LRC construction using the codes $C(sG, D)$. In this case note that for the same replication sets $A_0$, $A_1$ and $A_2$, the dual code of an LRC with locality $8$ is generated by evaluating $\{  L(T)^i \ |  \ 0 \leq i \leq 2 \}$. If we extend this to $\mathbf{F}_{27}^3$ we can get a code with locality $8$ and availability $3$ as the dual code of the evaluation of $\{  L(X)^iL(Y)^jL(Z)^k \ |  \ 0 \leq i \leq 2 \}$. In order to find subcodes of $C(sG, D)$ with  locality $8$ and availability $3$ we consider how many of the functions $\{  L(X)^iL(Y)^jL(Z)^k \ |  \ 0 \leq i \leq 2 \}$ also have weight $s$. The dimension of the LRC is found by computing the dimension of $\mathcal{L}(s) + \{  L(X)^iL(Y)^jL(Z)^k \ |  \ 0 \leq i \leq 2 \}$.

In this case we have found a $[19683,4536,600]$ LRC with locality $8$ and availability $3$ from the product code construction of the Reed--Solomon code and a $[19683,15434,600]$ LRC with locality $8$ and availability $3$ from $C(sG, D)^\perp$.

We have found instances in which the Reed--Solomon product codes are better than the LRC from the AG codes. Likewise, we have found cases in which the AG LRCs outperform the Reed--Solomon codes. Further improvements could be possible as knowledge of the Riemann-Roch spaces of the Ree curve improves.

\end{document}